\documentclass[conference, onecolumn, letterpaper]{IEEEtran}
\usepackage[utf8]{inputenc}
\usepackage{amsmath, amssymb, amsthm, bbm, mathtools, nccmath}
\usepackage{url, hyperref}
\usepackage{tikz,pgfplots}
\usetikzlibrary{positioning, decorations.pathreplacing}
\pgfplotsset{compat=1.18}
\theoremstyle{plain}
\newtheorem{theorem}{Theorem}
\newtheorem{lemma}[theorem]{Lemma}

\newtheorem{proposition}[theorem]{Proposition}

\theoremstyle{definition}

\theoremstyle{remark}

\renewcommand{\mathbf}{\boldsymbol}
\renewcommand{\tilde}{\widetilde}
\renewcommand{\leq}{\leqslant}
\renewcommand{\geq}{\geqslant}
\newcommand{\set}[1]{\mathcal{#1}}
\newcommand{\fnc}[1]{\mathrm{#1}}
\newcommand{\rv}[1]{\mathsf{#1}}
\newcommand{\sys}[1]{\mathsf{#1}}

\newcommand{\syss}[1]{\mathbf{\mathsf{#1}}}
\newcommand{\defeq}{\coloneqq}
\newcommand{\defas}{\eqqcolon}

\newcommand{\hilbert}{\mathcal{H}}
\newcommand{\bra}[1]{\left\lvert#1\right\rangle}

\newcommand{\braket}[1]{\left\lvert#1\middle\rangle\!\middle\langle#1\right\rvert}

\DeclareMathOperator{\ID}{id} 
\newcommand{\tensor}{\otimes}

\DeclareMathOperator{\tr}{tr}

\DeclarePairedDelimiterX{\infdiv}[2]{(}{)}{#1\delimsize\Vert#2}
\DeclarePairedDelimiterX{\inner}[2]{\langle}{\rangle}{#1,#2}
\newcommand{\tos}[2]{\stackrel{\mathclap{\small\mbox{#1}}}{#2}} 
\ExplSyntaxOn
\NewDocumentCommand{\multiadjustlimits}{m}
 {
  \group_begin:
  \multiadjustlimits_measure:n { #1 }
  \multiadjustlimits_print:n { #1 }
  \group_end:
 }

\tl_new:N  \l__multiadjustlimits_operator_tl
\tl_new:N  \l__multiadjustlimits_limit_tl

\cs_new_protected:Nn \multiadjustlimits_measure:n
 {
  \clist_map_function:nN { #1 } \__multiadjustlimits_measure:n
 }
\cs_new_protected:Nn \__multiadjustlimits_measure:n
 {
  \__multiadjustlimits_measure:NNn #1
 }
\cs_new_protected:Nn \__multiadjustlimits_measure:NNn
 {
  \tl_put_right:Nn \l__multiadjustlimits_operator_tl { #1 }
  \tl_put_right:Nn \l__multiadjustlimits_limit_tl { #3 }
 }

\cs_new_protected:Nn \multiadjustlimits_print:n
 {
  \clist_map_function:nN { #1 } \__multiadjustlimits_print:n
 }
\cs_new_protected:Nn \__multiadjustlimits_print:n
 {
  \__multiadjustlimits_print:NNn #1
 }
\cs_new_protected:Nn \__multiadjustlimits_print:NNn
 {
  \mathop { \vphantom{\l__multiadjustlimits_operator_tl} \mathopen{} #1 }
  \limits
  \sb{ \vphantom{\cramped{\l__multiadjustlimits_limit_tl}} #3 }
 }

\ExplSyntaxOff
\makeatletter
\newcommand\ie{\textit{i.e.}}
\newcommand\eg{\textit{e.g.}}

\newcommand\wrt{w.r.t.~}
\newcommand\aka{a.k.a.~}

\newcommand\Cf{\textit{cf.}}
\makeatother

\begin{document}
\title{Quantum Channel Simulation in Fidelity is no more difficult than State Splitting}
\author{\IEEEauthorblockN{Michael X. Cao\IEEEauthorrefmark{1},
    Rahul Jain\IEEEauthorrefmark{1}\IEEEauthorrefmark{3}\IEEEauthorrefmark{4},
    Marco Tomamichel\IEEEauthorrefmark{1}\IEEEauthorrefmark{2}}
    \IEEEauthorblockA{\IEEEauthorrefmark{1}Centre for Quantum Technologies, National University of Singapore, Singapore}
    \IEEEauthorblockA{\IEEEauthorrefmark{2}Department of Electrical and Computer Engineering, National University of Singapore, Singapore}
    \IEEEauthorblockA{\IEEEauthorrefmark{3}Department of Computer Science, National University of Singapore, Singapore}
    \IEEEauthorblockA{\IEEEauthorrefmark{4}MajuLab, International Joint Research Unit UMI 3654}
}
\maketitle
\begin{abstract}
Characterizing the minimal communication needed for the quantum channel simulation is a fundamental task in the quantum information theory.
In this paper, we show that, in fidelity, the quantum channel simulation can be directly achieved via quantum state splitting without using a technique known as the de~Finetti reduction, and thus provide a pair of tighter one-shot bounds.
Using the bounds, we also recover the quantum reverse Shannon theorem in a much simpler way.
\end{abstract}
\section{Introduction}
We consider the problem of simulating a quantum channel using entanglement-assisted local operations and classical communications (eLOCC).
We are interested in characterizing the minimal classical communication necessary for a faithful simulation of the channel measured in fidelity.
This is a fundamental task in quantum information theory, and the first-order asymptotic rate of the minimal classical communication is characterized by the entanglement-assisted capacity of the target channel, which is known as the \emph{reverse Shannon theorem}~\cite{bennett2002entanglement, bennett2014quantum}.
Recent years have seen a number of studies of the problem in different regimes, including the one-shot no-signaling-assisted regime~\cite{fang2019quantum}, the moderate deviation regime~\cite{ramakrishnan2023moderate}, and network setups~\cite{kurri2022multiple, cheng2023quantum}.

However, despite the recent development, it remains an open task to characterize the asymptotic minimal rate of communication for quantum channel simulation in the second order.
One of the major difficulties lies within the requirement that a channel simulation protocol must work for all input states simultaneously.
This is in stark contrast with a highly related task known as the \emph{quantum state splitting} (more precisely, a special case of the task known as the quantum state transfer).
In particular, in both~\cite{ramakrishnan2023moderate} and~\cite{cheng2023quantum}, the authors approached the problem of quantum channel simulation via the quantum state splitting of some so-called \emph{de~Finetti state}, at the cost of a multiplier before the deviation term $\epsilon$ that grows polynomially \wrt the blocklength $n$ (see, \eg,~\cite[Eq.~(105)]{ramakrishnan2023moderate}).
This makes further studies of higher-order analyses very difficult along the same approach, if not impossible.

In this paper, we provide a much more direct relationship between the task of quantum state splitting and the quantum channel simulation.
In particular, we show that the fidelity between the joint input-output density operators of the target channel and that of the simulated channel (see~\eqref{eq:channel:distance}) is convex \wrt the input density operator while concave \wrt to the protocol (as a CPTP map).
Using Sion's minimax theorem, this implies that the protocol that works best for the worst input density operator has the same performance as the worst one among the protocols optimized for each input density operator (see~\eqref{eq:minimax}).
This finding not only provides a tighter one-shot achievability bound (\Cf~\cite{ramakrishnan2023moderate}), but also leads to a much simpler proof of the reverse Shannon theorem.
Moreover, this opens up new possibilities for further studies on higher-order analyses of this problem.

In the following part of the paper, we first introduce the problem of quantum channel simulation and quantum state splitting together with suitable notations.
Second, we show a direct connection between the two tasks, and thus provide a pair of tighter one-shot upper and lower bounds on the minimal message size for simulating a quantum channel with fidelity at least $1-\epsilon^2$.
Lastly, we recover the first-order asymptotic results, \aka the quantum reverse Shannon theorem, using the newly found upper and lower bounds in a much simpler way. 

\section{Quantum Channel Simulation and Quantum State Splitting}

We hereby describe the task of simulating finite-dimensional quantum channels using entanglement-assisted local operations and classical communication.
Suppose that we are given a quantum channel from system $\sys{A}$ to $\sys{B}$ described by some completely-positive-trace-preserving (CPTP) map $\mathcal{N}_{\sys{A}\to\sys{B}}:\set{D}(\hilbert_\sys{A})\to\set{D}(\hilbert_\sys{B})$ where the state spaces $\hilbert_\sys{A}$ and $\hilbert_\sys{B}$ are both finite-dimensional Hilbert spaces.
We would like to find
\begin{itemize}
    \item a pair of entangled systems $\sys{K}'$ and $\sys{K}$ (with their joint state being some pure state $\bra{\sigma}_{\sys{KK}'}$),
    \item (Alice) a joint local measurement on systems $\sys{A}$ and $\sys{K'}$ (described by some POVM $\{E_m\}_{m\in[M]}$),
    \item (Bob) a local operation from system $\sys{K}$ to $\sys{B}$ (described by some classical-controlled CPTP map $\Phi^{(m)}_{\sys{K}\to\sys{B}}$),
\end{itemize}
such that the joint effect of the latter two operations (which is effectively a CPTP map from system $\sys{A}$ to $\sys{B}$), \ie,
\begin{equation*}
   \tilde{\mathcal{N}}_{\sys{A}\to\sys{B}}: \rho_\sys{A} \mapsto \sum_{m\in[M]} \Phi_{\sys{K}\to\sys{B}}^{(m)}\left(\tr_{\sys{AK'}}\left[(E_m\tensor I_{\sys{K}})\cdot (\rho_\sys{A}\tensor \braket{\sigma}_{\sys{KK}'})\right]\right)
\end{equation*}
``resembles'' the channel $\mathcal{N}_{\sys{A}\to\sys{B}}$.
This process is depicted in Figure~\ref{fig:channel:simulation}.
\begin{figure}
\centering
\includegraphics{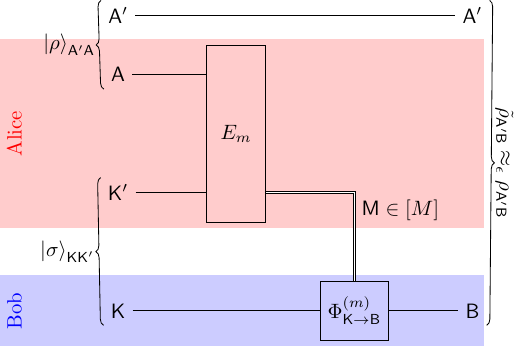}
\caption{The task of quantum channel simulation with fidelity at least $1-\epsilon^2$. The goal is to have $\tilde{\rho}_{\sys{A}'\sys{B}} \approx_{\epsilon}\rho_{\sys{A}'\sys{B}} \defeq\ID_{\sys{A}'}\tensor\mathcal{N}_{\sys{A}\to\sys{B}}(\braket{\rho}_{\sys{A}'\sys{A}})$ \emph{for all} input states $\rho_{\sys{A}'}$, where $\braket{\rho}_{\sys{A}'\sys{A}}$ is the canonical purification of $\rho_{\sys{A}'}$.}    
\label{fig:channel:simulation}
\end{figure}
More precisely, we are interested in finding the minimal alphabet size $M$ such that \emph{for all} $\rho_{\sys{A}'}\in\set{D}(\hilbert_{\sys{A}'})$
\begin{equation}\label{eq:channel:distance}
    f(\tilde{\mathcal{N}}_{\sys{A}\to\sys{B}}, \rho_{\sys{A}'}) \defeq \sqrt{ F(
    \underbrace{\ID_{\sys{A}'}\tensor\mathcal{N}_{\sys{A}\to\sys{B}}(\braket{\rho}_{\sys{A}'\sys{A}})}_{\defas \rho_{\sys{A}'\sys{B}}},
    \underbrace{\ID_{\sys{A}'}\tensor\tilde{\mathcal{N}}_{\sys{A}\to\sys{B}}(\braket{\rho}_{\sys{A}'\sys{A}}))}_{\defas \tilde{\rho}_{\sys{A}'\sys{B}}} }
    \geq \sqrt{1-\epsilon^2}
\end{equation}
for some given $\epsilon\in(0,1)$.
Here, the quantum systems $\sys{A}$ and $\sys{A}'$ have the same state space, and $\braket{\rho}_{\sys{A}'\sys{A}}\defeq(\sqrt{\rho_{\sys{A}'}}\tensor I_{\sys{A}}) \braket{\gamma}(\sqrt{\rho_{\sys{A}'}}\tensor I_\sys{A})$ is the canonical purification of $\rho_{\sys{A}'}$ on $\sys{A}$ where $\bra{\gamma}$ is the maximal entangled state on the joint system $\sys{A}'\sys{A}$.
We use the following definition for the fidelity
\begin{equation*}
F(\rho,\sigma) \defeq \left(\tr\sqrt{\rho^{1/2}\sigma\rho^{1/2}}\right)^2	.
\end{equation*}

On the other hand, quantum state splitting is a highly related task.
In particular, quantum channel simulation can be seen as a ``\emph{universal}'' version of the quantum state transfer, and the latter is a special case of quantum state splitting.
Given some composite system $\sys{SP}$ with its state described by some known fixed density operator $\rho_\sys{SP}$, the task of quantum sate splitting is to send $\sys{P}$ from Alice to Bob using (one-way) classical communication and entanglement-assisted local operations, where at the beginning of the protocol Alice has access to both $\sys{S}$ and $\sys{P}$, and at the end of the protocol Bob has access to $\sys{P}$, and the state of $\sys{RSP}$, described by $\tilde{\rho}_\sys{RSP}$, is close to $\braket{\rho}_\sys{RSP}$ in fidelity.
Here, $\sys{R}$ is some reference system that purifies $\sys{SP}$.
This is illustrated in Fig.~\ref{fig:state:splitting}.
\begin{figure}
\centering
\includegraphics{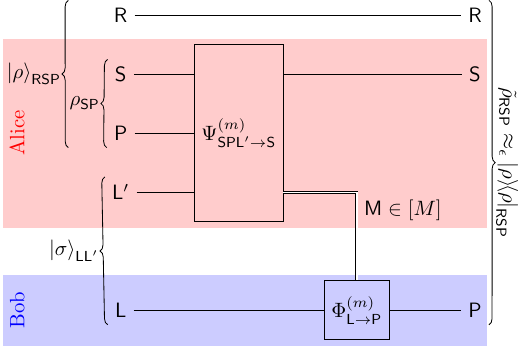}
\caption{The task of quantum state splitting with high fidelity (at least $1-\epsilon^2$).
The goal is to have $\tilde{\rho}_\sys{RSP}\approx_\epsilon \braket{\rho}_\sys{RSP}$ where $\rho_{\sys{SP}}$ is fixed and known prior to the operations, and $\sys{R}$ is some reference system purifying $\sys{SP}$.}
\label{fig:state:splitting}
\end{figure}
The major difference between the two tasks is that the protocols for the state splitting are $\rho_\sys{SP}$-specific; whereas the protocols for channel simulation have to work for all possible $\rho_{\sys{A}'}$ with no knowledge or assumptions of it.
In particular, the quantum channel simulation can be achieved by some \emph{universal} state splitting protocol, \ie, a state splitting protocol that works for all possible $\rho_\sys{SP}$ (see Fig.~\ref{fig:simulation:splitting}).
Without the ``universality'' of the state splitting protocol, assuming, for example, that we simply choose the best state splitting protocol for $\rho_\sys{EB}\defeq U_{\sys{A}\to\sys{EB}}\cdot \rho_\sys{A} \cdot U_{\sys{A}\to\sys{EB}}^\dagger$, the protocol in Fig.~\ref{fig:simulation:splitting} only gives rise to a protocol as in Fig.~\ref{fig:channel:simulation} that only works for this specific input.
\begin{figure}
\centering
\includegraphics{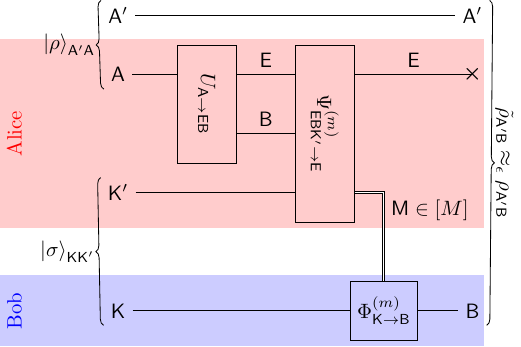}
\caption{A quantum channel simulation protocol constructed  from a state splitting protocol.
   Here, $U_{\sys{A}\to\sys{EB}}$ is the isometry representation of the original channel $\mathcal{N}_{\sys{A}\to\sys{B}}$.
   Note that we used the state splitting protocol on systems $\sys{E}$ and $\sys{B}$, and then discarded system $\sys{E}$.}    
\label{fig:simulation:splitting}
\end{figure}
For this very purpose, in the previous work~\cite{ramakrishnan2023moderate}, the state splitting protocols on the de Finetti state was considered when studying quantum channel simulations.

\section{Quantum Channel Simulation via State Splitting}

In this section, we show that the expression in~\eqref{eq:channel:distance} is concave in $\tilde{\mathcal{N}}_{\sys{A}\to\sys{B}}$ and quasi-convex in $\rho_\sys{A}$.
This allows us to apply the Sion's minimax theorem\footnote{Together with the facts that the set $\set{D}(\hilbert_{\sys{A}'})$ is convex and closed, and that the set $\mathfrak{P}^{(M)}_{\sys{A}\to\sys{B}}$ is convex.}, and write
\begin{equation}\label{eq:minimax}
	\adjustlimits
	\sup_{\tilde{\mathcal{N}}_{\sys{A}\to\sys{B}}\in\mathfrak{P}^{(M)}_{\sys{A}\to\sys{B}}}
	\inf_{\rho_{\sys{A}'}\in\set{D}(\hilbert_{\sys{A}'})}
         f(\tilde{\mathcal{N}}_{\sys{A}\to\sys{B}}, \rho_{\sys{A}'})
	= \adjustlimits
	\inf_{\rho_{\sys{A}'}\in\set{D}(\hilbert_{\sys{A}'})}
	\sup_{\tilde{\mathcal{N}}_{\sys{A}\to\sys{B}}\in\mathfrak{P}^{(M)}_{\sys{A}\to\sys{B}}}
         f(\tilde{\mathcal{N}}_{\sys{A}\to\sys{B}}, \rho_{\sys{A}'})
\end{equation}
where $\mathfrak{P}^{(M)}_{\sys{A}\to\sys{B}}$ is the set of all eLOCC protocols with alphabet size $M$ (formally defined below in~\eqref{eq:def:set:E}).
In other words, under the same communication constraint, the best protocol for channel simulation has the same performance as the worst-performing protocol among the best protocols for each $\rho_{\sys{A}'}$.
This allows us to use the protocols derived from the state-splitting protocols (as in Fig.~\ref{fig:simulation:splitting}) and its achievability bounds (see~\cite[Theorem~3]{ramakrishnan2023moderate} and~\cite[Theorem~1]{anshu2017quantum}) to provide a one-shot achievability bound for the channel simulation.
It is worth-noting that there are achievability bounds in network communication tasks that utilize the Sion's minimax theorem in similar ways (\eg, see~\cite{anshu2020noisy} and~\cite{anshu2018hypothesis}).
This bound matches with the converse bound (with small fudge terms) one can derive using the non-lockability property and the data-processing inequality of max-mutual information (\eg, see~\cite[Proposition~32]{ramakrishnan2023moderate}).

We formalize the set of all eLOCC protocols as described at the beginning of this paper.
Given quantum systems $\sys{A}$ and $\sys{B}$, we denote $\mathfrak{C}_{\sys{A}\to\sys{B}}$ the set of CPTP maps from $\sys{A}$ to $\sys{B}$, and we define the set of \emph{entanglement-assisted local-operation classical-communication (eLOCC) protocols from $\sys{A}$ to $\sys{B}$ with alphabet size $M\in\mathbb{N}$} as a subset of $\mathfrak{C}_{\sys{A}\to\sys{B}}$ as
\begin{equation}\label{eq:def:set:E}
\mathfrak{P}_{\sys{A}\to\sys{B}}^{(M)} \defeq \left\{
\begin{aligned}
&\tilde{\mathcal{N}}_{\sys{A}\to\sys{B}}: \set{D}(\hilbert_\sys{A}) \to \set{D}(\hilbert_\sys{B})\\
&\rho_\sys{A} \mapsto \sum_{m\in[M]} \Phi_{\sys{K}\to\sys{B}}^{(m)}\left(\tr_{\sys{AK'}}\left[(E_m\tensor I_{\sys{K}})\cdot (\rho_\sys{A}\tensor \braket{\sigma}_{\sys{KK}'})\right]\right)
\end{aligned}
\middle\vert\ 
\begin{aligned}
    & \sys{K}\text{, }\sys{K}'\text{ are quantum systems with }\hilbert_\sys{K}=\hilbert_{\sys{K}}'\\
    & \{E_m\}_{m\in[M]}\text{ is some POVM on the joint system } \sys{AK}'\\
    & \Phi^{(m)}_{\sys{K}\to\sys{B}} \text{ is some CPTP from }\sys{K} \text{ to }\sys{B} \text{ for each }m
\end{aligned}
\right\}.
\end{equation}
Notice that $\mathfrak{P}_{\sys{A}\to\sys{B}}^{(M)}$ is a convex (but not closed) subset of $\mathfrak{C}_{\sys{A}\to\sys{B}}$.
To see $\mathfrak{P}_{\sys{A}\to\sys{B}}^{(M)}$ to be convex, we observe that any convex combination of two eLOCC protocols can be achieved using a single bit of shared randomness, \ie, Alice and Bob can choose to use protocol \#1 if the bit turns out to be `0', or protocol \#2 if the bit is `1'.
The shared randomness can be extracted from a pair of entangled qubits; and the latter can be provided by enlarging the dimensions of the systems $\sys{K}$ and $\sys{K}'$.

For a given quantum channel $\mathcal{N}_{\sys{A}\to\sys{B}}$ from system $\sys{A}$ to system $\sys{B}$, the best performance (in terms of fidelity) of all $M$-alphabet-size eLOCC protocols for simulating $\mathcal{N}_{\sys{A}\to\sys{B}}$ can be expressed as 
\begin{equation}\label{eq:def:f}
    1-(\epsilon^\star_M)^2 = 
    \left(\adjustlimits
    \sup_{\tilde{\mathcal{N}}_{\sys{A}\to\sys{B}}\in\mathfrak{P}^{(M)}_{\sys{A}\to\sys{B}}}
	\inf_{\rho_{\sys{A}'}\in\set{D}(\hilbert_{\sys{A}'})}
         f(\tilde{\mathcal{N}}_{\sys{A}\to\sys{B}}, \rho_{\sys{A}'}) \right)^2.
\end{equation}
Recall that $\hilbert_{\sys{A}'}=\hilbert_\sys{A}$, and $\braket{\rho}_{\sys{A}'\sys{A}}\defeq (\sqrt{\rho_{\sys{A}'}}\tensor I_\sys{A})\braket{\gamma}(\sqrt{\rho_{\sys{A}'}}\tensor I_\sys{A})$ is the canonical purification of $\rho_{\sys{A}'}$ on $\sys{A}$ where $\bra{\gamma}$ is the maximal entangled state on the joint system $\sys{AA}'$.
We consider $f$ in~\eqref{eq:channel:distance} as a function defined on $\mathfrak{P}^{(M)}_{\sys{A}\to\sys{B}} \times \set{D}(\hilbert_{\sys{A}'})$.
Since the fidelity is a jointly concave function, the function $f$ is also concave in its first argument $\tilde{\mathcal{N}}_{\sys{A}\to\sys{B}}\in\mathfrak{P}^{(M)}_{\sys{A}\to\sys{B}}$ for each fixed $\rho_{\sys{A}'}$.
In the following, we show that the function $f$ is convex \wrt to its second argument $\rho_{\sys{A}'}\in\set{D}(\hilbert_{\sys{A}'})$.
\begin{lemma}\label{lemma:f:convex}
    The function $f$ defined in~\eqref{eq:channel:distance} is convex in $\rho_{\sys{A}'}\in\set{D}(\hilbert_{\sys{A}'})$ for each fixed $\tilde{\mathcal{N}}_{\sys{A}\to\sys{B}}\in\mathfrak{P}^{(M)}_{\sys{A}\to\sys{B}}$.
\end{lemma}
This can be shown as a direct result of~\cite[Proposition~4.80]{khatri2020principles}.
However, for completeness, we provide a short proof as follows.
(An alternative proof is also included in Appendix~\ref{app:proof:lemma:f:convex}.)
\begin{proof}
Let $\tilde{\mathcal{N}}_{\sys{A}\to\sys{B}} \in\mathfrak{P}^{(M)}_{\sys{A}\to\sys{B}}$ be fixed, and let $\rho_0\neq\rho_1\in\set{D}(\hilbert_{\sys{A}'})$ be picked arbitrarily.
For each $\lambda\in[0,1]$, denote $\rho_\lambda \defeq (1-\lambda)\cdot\rho_0 + \lambda\cdot \rho_1$.
Also denote $\bra{\rho_\lambda}\defeq (\sqrt{\rho_\lambda}\tensor I)\braket{\gamma}(\sqrt{\rho_\lambda}\tensor I)$ the canonical purification of $\rho_\lambda$.
Note that $\bra{\tilde{\rho}_\lambda}\defeq \sqrt{1-\lambda}\bra{0}\bra{\rho_0}+\sqrt{\lambda}\bra{1}\bra{\rho_1}$ is also a purification of $\rho_\lambda$.
Therefore, denoting $\sys{R}$ a single-qubit auxiliary system, we have
\begin{align}
f(\tilde{\mathcal{N}}_{\sys{A}\to\sys{B}}, \rho_\lambda)
& = F^{1/2}(\ID_{\sys{A}'}\tensor\mathcal{N}_{\sys{A}\to\sys{B}}(\braket{\rho_\lambda}), \ID_{\sys{A}'}\tensor\tilde{\mathcal{N}}_{\sys{A}\to\sys{B}}(\braket{\rho_\lambda})) \nonumber \\
& = F^{1/2}(\ID_\sys{R}\tensor\ID_{\sys{A}'}\tensor\mathcal{N}_{\sys{A}\to\sys{B}}(\braket{\tilde{\rho}_\lambda}), \ID_\sys{R}\tensor\ID_{\sys{A}'}\tensor\tilde{\mathcal{N}}_{\sys{A}\to\sys{B}}(\braket{\tilde{\rho}_\lambda})) \label{lemma:f:convex:1} \\
& \leq \begin{aligned}[t]
F^{1/2}(\ID_\sys{R}\tensor\ID_{\sys{A}'}\tensor\mathcal{N}_{\sys{A}\to\sys{B}}((1-\lambda)\braket{0}\tensor\braket{\rho_0}+\lambda\braket{1}\tensor\braket{\rho_1}), \ldots \\
\ID_\sys{R}\tensor\ID_{\sys{A}'}\tensor\tilde{\mathcal{N}}_{\sys{A}\to\sys{B}}((1-\lambda)\braket{0}\tensor\braket{\rho_0}+\lambda\braket{1}\tensor\braket{\rho_1}))
\end{aligned} \label{lemma:f:convex:2}\\
& = \begin{aligned}[t]
(1-\lambda)\cdot F^{1/2}(\ID_{\sys{A}'}\tensor\mathcal{N}_{\sys{A}\to\sys{B}}(\braket{\rho_0}), \ID_{\sys{A}'}\tensor\tilde{\mathcal{N}}_{\sys{A}\to\sys{B}}(\braket{\rho_0})) + \ldots \\
\lambda\cdot F^{1/2}(\ID_{\sys{A}'}\tensor\mathcal{N}_{\sys{A}\to\sys{B}}(\braket{\rho_1}), \ID_{\sys{A}'}\tensor\tilde{\mathcal{N}}_{\sys{A}\to\sys{B}}(\braket{\rho_1})) 
\end{aligned}\nonumber \\
&= (1-\lambda)\cdot f(\tilde{\mathcal{N}}_{\sys{A}\to\sys{B}}, \rho_0) + \lambda\cdot f(\tilde{\mathcal{N}}_{\sys{A}\to\sys{B}}, \rho_1), \nonumber
\end{align} 
where we use the Uhlmann's theorem in~\eqref{lemma:f:convex:1}, and measured the system $\sys{R}$ in~\eqref{lemma:f:convex:2}.
\end{proof}

Lemma~\ref{lemma:f:convex} provides a direct connection between the task of channel simulation and the state splitting.
In particular, since the set $\set{D}(\hilbert_{\sys{A}'})$ is closed and convex, and the set $\mathfrak{P}^{(M)}_{\sys{A}\to\sys{B}}$ is convex, we can apply the Sion's minimax theorem, \ie,
\begin{theorem}[Sion's minimax theorem~{\cite{sion1958general}}]
Let $\set{X}$ be a compact convex set and $\set{Y}$ be a convex set.
If a function $f:\set{X}\times\set{Y}\to\mathbb{R}$ satiesfies
\begin{itemize}
    \item $f(x,\cdot)$ is upper semi-continuous and quasi-concave on $\set{Y}$ for each fixed $x\in\set{X}$,
    \item $f(\cdot,y)$ is lower semi-continuous and quasi-convex on $\set{X}$ for each fixed $y\in\set{Y}$,
\end{itemize}
then,
\begin{equation*}
\adjustlimits \min_{x\in\set{X}} \sup_{y\in\set{Y}} f(x,y) = \adjustlimits \sup_{y\in\set{Y}} \min_{x\in\set{X}} f(x,y).
\end{equation*}
\end{theorem}
Using the above theorem, we rewrite~\eqref{eq:def:f} as
\begin{equation}\label{eq:minimax:1}
    1-(\epsilon^\star_M)^2 = 
 \left(\adjustlimits
     \inf_{\rho_{\sys{A}'}\in\set{D}(\hilbert_{\sys{A}'})}
    \sup_{\tilde{\mathcal{N}}_{\sys{A}\to\sys{B}}\in\mathfrak{P}^{(M)}_{\sys{A}\to\sys{B}}}
         f(\tilde{\mathcal{N}}_{\sys{A}\to\sys{B}}, \rho_{\sys{A}'}) \right)^2.
\end{equation}
In other words, the optimal performance of channel simulations is directly determined by the optimal performance of quantum state transfers using eLOCC protocols under the same classical communication constraint.
The latter can be achieved using quantum state-splitting protocols (see Fig.~\ref{fig:simulation:splitting}) provided that the message size $M$ is large enough~\cite{ramakrishnan2023moderate}.
\begin{proposition}\label{prop:one-shot:achivability}
Given a quantum channel $\mathcal{N}_{\sys{A}\to\sys{B}}$, there exists an eLOCC protocol with alphabet size $M$ that simulates $\mathcal{N}_{\sys{A}\to\sys{B}}$ with fidelity at least $1-\epsilon^2$ if
\begin{equation}\label{eq:one-shot:achivability}
\log{M} \geq \adjustlimits \sup_{\rho_{\sys{A}'}\in\set{D}(\hilbert_{\sys{A}'})} \inf_{\sigma_\sys{B}\in\set{D}(\hilbert_\sys{B})} D_{\max}^{\epsilon-\delta, \sys{A}'}\infdiv*{\rho_{\sys{A}'\sys{B}}}{\rho_{\sys{A}'}\tensor\sigma_\sys{B}} -\log{\delta^2},
\end{equation}
for some $\delta\in(0,\epsilon)$, where $\rho_{\sys{A}'\sys{B}} \defeq \ID_{\sys{A}'} \tensor \mathcal{N}_{\sys{A}\to\sys{B}}(\braket{\rho}_{\sys{A}'\sys{A}})$.
Here, for density operators $\varrho_{\sys{AB}}, \varsigma_{\sys{AB}}\in\set{D}(\hilbert_\sys{AB})$, and $\varepsilon\in(0,1)$, the partial smoothed max-divergence $D_{\max}^{\varepsilon,\sys{A}}\infdiv*{\varrho_{\sys{AB}}}{\varsigma_\sys{AB}}$ is defined as
\begin{equation*}
D_{\max}^{\varepsilon,\sys{A}}\infdiv*{\varrho_{\sys{AB}}}{\varsigma_\sys{AB}} \defeq \inf_{\tilde{\varrho}_{\sys{AB}}\in\set{D}(\hilbert_\sys{AB}): \tilde{\varrho}_\sys{A}=\varrho_\sys{A}, F(\tilde{\varrho}_{\sys{AB}},\varrho_{\sys{AB}})\geq1-\varepsilon^2} D_{\max}\infdiv*{\tilde{\varrho}_{\sys{AB}}}{\varsigma_{\sys{AB}}}.
\end{equation*}
\end{proposition}
\begin{proof}
This is a direct consequence of~\eqref{eq:minimax:1} and the results on the quantum state splitting (see~\cite[Theorem~3]{ramakrishnan2023moderate} and~\cite[Theorem~1]{anshu2017quantum}), \ie, given a pure state $\bra{\rho}_{\sys{A}'\sys{EB}}$, there exist a quantum state splitting protocol on systems $\sys{E}$ and $\sys{B}$ that achieves the $(1-\epsilon^2)$-fidelity if
\begin{equation}\label{eq:state:split:achievability}
    \log{M} \geq \inf_{\sigma_\sys{B}\in\set{D}(\hilbert_\sys{B})} D_{\max}^{\epsilon-\delta,\sys{A}'}\infdiv*{\rho_{\sys{A}'\sys{B}}}{\rho_{\sys{A}'}\tensor\sigma_\sys{B}} - \log{\delta^2}.
\end{equation}
In other words, for an integer $M$ large enough such that~\eqref{eq:one-shot:achivability} holds for \emph{all} $\rho_{\sys{A}'}\in\set{D}(\hilbert_{\sys{A}'})$, using the quantum state splitting protocol guaranteed to exists above, one can construct an eLOCC protocol $\tilde{\mathcal{N}}_{\sys{A}\to\sys{B}}$ \emph{for each $\rho_{\sys{A}'}$} such that (see~Fig.~\ref{fig:simulation:splitting}) 
\begin{equation*}
f(\tilde{\mathcal{N}}_{\sys{A}\to\sys{B}}, \rho_{\sys{A}'}) \geq \sqrt{1-\epsilon^2}.
\end{equation*}
Referring to~\eqref{eq:minimax:1}, the maximum fidelity that can be achieved by $M$-alphabet eLOCC protocols is at least
\begin{equation*}
1-(\epsilon_M^\star)^2 \geq \inf_{\rho_{\sys{A}'}\in\set{D}(\hilbert_{\sys{A}'})} 1-\epsilon^2 = 1-\epsilon^2.
\end{equation*}
Thus, there must exists at least one such protocol that simulates $\mathcal{N}_{\sys{A}\to\sys{B}}$ with fidelity at least $1-\epsilon^2$. 
\end{proof}

Similar to \cite[Proposition~32]{ramakrishnan2023moderate} and \cite[Theorem~2]{anshu2017quantum}, we have the following one-shot converse bound.
\begin{proposition}\label{prop:one-shot:converse}
Given a quantum channel $\mathcal{N}_{\sys{A}\to\sys{B}}$, for any $M$-alphabet-size eLOCC protocols that simulates $\mathcal{N}_{\sys{A}\to\sys{B}}$ with fidelity at least $1-\epsilon^2$, it holds that
\begin{equation}\label{eq:one-shot:converse}
\log{M} \geq \adjustlimits \sup_{\rho_{\sys{A}'}\in\set{D}(\hilbert_{\sys{A}'})} \inf_{\sigma_\sys{B}\in\set{D}(\hilbert_\sys{B})} D_{\max}^{\epsilon,\sys{A}'}\infdiv*{\rho_{\sys{A}'\sys{B}}}{\rho_{\sys{A}'}\tensor\sigma_\sys{B}}.
\end{equation}
Recall that $\rho_{\sys{A}'\sys{B}} \defeq \ID_{\sys{A}'} \tensor \mathcal{N}_{\sys{A}\to\sys{B}}(\braket{\rho}_{\sys{A}'\sys{A}})$.
\end{proposition}
\begin{proof}
Suppose we have an eLOCC protocol with alphabet size $M$ that simulates the channel $\mathcal{N}_{\sys{A}\to\sys{B}}$.
Let $\rv{M}$ denote the random variable representing the classical message (see Fig.~\ref{fig:channel:simulation}).
Starting from the picture of the systems right after the classical message $\rv{M}$ is gererated, we have the following chain of ineqalities
\begin{align}
\log{M} & = I_{\max}(\sys{A}':\sys{K}) + \log{M} \nonumber \\
&\geq I_{\max}(\sys{A}':\rv{M}\sys{K}) \label{eq:one-shot:converse:1} \\
&\geq I_{\max}(\sys{A}':\sys{B})_{\tilde{\rho}_{\sys{A}'\sys{B}}}, \label{eq:one-shot:converse:2}
\end{align}
where we used non-lockability of $I_{\max}$ (see~\cite[Cor.~A.14]{berta2013quantum}) in~\eqref{eq:one-shot:converse:1}, and the data-processing inequality of $I_{\max}$ in~\eqref{eq:one-shot:converse:2}, and we denote the density operator for systems $\sys{A}'\sys{B}$ at the end of the protocol by $\tilde{\rho}_{\sys{A}'\sys{B}}$.
Using the definition of $I_{\max}$, we have
\begin{align*}
I_{\max}(\sys{A}':\sys{B})_{\tilde{\rho}_{\sys{A}'\sys{B}}}  &= \inf_{\sigma_\sys{B}\in\set{D}(\hilbert_\sys{B})} D_{\max}\infdiv*{\tilde{\rho}_{\sys{A}'\sys{B}}}{\rho_{\sys{A}'}\tensor \sigma_\sys{B}} \\
&\geq \inf_{\sigma_\sys{B}\in\set{D}(\hilbert_\sys{B})} D_{\max}^{\epsilon,\sys{A}'}\infdiv*{\rho_{\sys{A}'\sys{B}}}{\rho_{\sys{A}'}\tensor\sigma_\sys{B}},
\end{align*}
where the last inequality is due to the hypothesis that the fidelity between $\tilde{\rho}_{\sys{A}'\sys{B}}$ and $\rho_{\sys{A}'\sys{B}}$ is at least $1-\epsilon^2$.
Combining the above, we know
\begin{equation*}
\log{M} \geq \inf_{\sigma_\sys{B}\in\set{D}(\hilbert_\sys{B})} D_{\max}^{\epsilon,\sys{A}'}\infdiv*{\rho_{\sys{A}'\sys{B}}}{\rho_{\sys{A}'}\tensor\sigma_\sys{B}}
\end{equation*}
for \emph{any} $\rho_{\sys{A}'}\in\set{D}(\hilbert_{\sys{A}'})$, which finishes the proof.
\end{proof}

\section{First-Order Analysis}
We now turn our attention to the asymptotic analysis of~\eqref{eq:one-shot:achivability} and~\eqref{eq:one-shot:converse}, \ie, the problem of simulating $n$ copies of the channel $\mathcal{N}_{\sys{A}\to\sys{B}}$.
Note that this problem has already been solved as the quantum reverse Shannon theorem~\cite{bennett2002entanglement, bennett2014quantum}, and we are merely recovering the result in a much simpler way.

For a fixed $\epsilon\in(0,1)$, let $M_\epsilon^\star(\mathcal{N}_{\sys{A}\to\sys{B}})$ denote the smallest alphabet size such that an eLOCC protocol can simulate $\mathcal{N}_{\sys{A}\to\sys{B}}$ with fidelity at least $1-\epsilon^2$.
We consider the asymptotics of the achievability bound first.
Starting by applying Proposition~\ref{prop:one-shot:achivability} on $n$ copies of $\mathcal{N}_{\sys{A}\to\sys{B}}$, we have 
\begin{align}
&\frac{1}{n}\log{M_\epsilon^\star(\mathcal{N}_{\sys{A}\to\sys{B}}^{\tensor n})}
 \leq \frac{1}{n}\adjustlimits \sup_{\rho_{{\syss{A}'}_1^n}\in\set{D}(\hilbert_{{\syss{A}'}_1^n})} \inf_{\sigma_{\syss{B}_1^n}\in\set{D}(\hilbert_{\syss{B}_1^n})} D_{\max}^{\epsilon-\delta, {\syss{A}'}_1^n}\infdiv*{\rho_{{\syss{A}'}_1^n\syss{B}_1^n}}{\rho_{{\syss{A}'}_1^n}\tensor\sigma_{\syss{B}_1^n}} -\frac{1}{n}\log{\delta^2} \nonumber \\
\label{eq:asymptotic:achivability:1}
&\hspace{25pt} \leq \frac{1}{n}\adjustlimits \sup_{\rho_{{\syss{A}'}_1^n}\in\set{D}(\hilbert_{{\syss{A}'}_1^n})} \inf_{\sigma_{\syss{B}_1^n}\in\set{D}(\hilbert_{\syss{B}_1^n})} D_{\max}^{\frac{\epsilon-\delta-\delta'}{2}}\infdiv*{\rho_{{\syss{A}'}_1^n\syss{B}_1^n}}{\rho_{{\syss{A}'}_1^n}\tensor\sigma_{\syss{B}_1^n}} + \frac{1}{n}\log{\frac{8+\delta'^2}{\delta'^2}} -\frac{1}{n}\log{\delta^2} \\
\label{eq:asymptotic:achivability:2}
&\hspace{25pt}  = \frac{1}{n}\adjustlimits \sup_{\rho_{{\syss{A}'}_1^n}\in\set{D}(\hilbert_{{\syss{A}'}_1^n})} \inf_{\sigma_{\syss{B}_1^n}\in\set{D}(\hilbert_{\syss{B}_1^n})} D_{\max}^{\epsilon/4}\infdiv*{\rho_{{\syss{A}'}_1^n\syss{B}_1^n}}{\rho_{{\syss{A}'}_1^n}\tensor\sigma_{\syss{B}_1^n}} + \frac{1}{n}\log{\frac{128+\epsilon^2}{\epsilon^2}} -\frac{1}{n}\log{\frac{\epsilon^2}{16}} \\
\label{eq:asymptotic:achivability:3}
&\hspace{25pt} \leq \frac{1}{n} \underbrace{\adjustlimits \sup_{\rho_{{\syss{A}'}_1^n}\in\set{D}(\hilbert_{{\syss{A}'}_1^n})} \inf_{\sigma_{\syss{B}_1^n}\in\set{D}(\hilbert_{\syss{B}_1^n})} \tilde{D}_{\alpha}\infdiv*{\rho_{{\syss{A}'}_1^n\syss{B}_1^n}}{\rho_{{\syss{A}'}_1^n}\tensor\sigma_{\syss{B}_1^n}}}_{\defas \tilde{I}_\alpha(\mathcal{N}_{\sys{A}\to\sys{B}}^{\tensor n})} + \underbrace{\frac{1}{n} \frac{-\log\left(1\!-\!\sqrt{1\!-\!\frac{\epsilon^2}{16}}\right)}{\alpha-1} + \frac{1}{n}\log{\frac{128+\epsilon^2}{\epsilon^2}} -\frac{1}{n}\log{\frac{\epsilon^2}{16}}}_{\to 0\text{ as } n\to\infty},\!\!
\end{align}
where in~\eqref{eq:asymptotic:achivability:1} we use~\cite[Theorem~11]{anshu2020partially},
in~\eqref{eq:asymptotic:achivability:2} we substitute $\delta,\delta'\gets\epsilon/4$,
in~\eqref{eq:asymptotic:achivability:3} we use~\cite[Proposition~6.5]{tomamichel2015quantum}.
Note that the sandwiched Rényi relative entropy is defined as
\begin{equation*}
\tilde{D}_\alpha\infdiv*{\rho}{\sigma} \defeq \frac{1}{\alpha-1} \log{\tr\left(\sigma^{\frac{1-\alpha}{2\alpha}}\rho\sigma^{\frac{1-\alpha}{2\alpha}}\right)^\alpha}.
\end{equation*}
The first part of~\eqref{eq:asymptotic:achivability:3} is the sandwiched Rényi mutual information of the channel $\mathcal{N}_{\sys{A}\to\sys{B}}^{\tensor n}$ which is known to be additive~\cite[Lemma~6]{gupta2015multiplicativity}; whereas the second part tends to zero as $n$ tends to infinity for any fixed $\alpha>1$ and $\epsilon\in(0,1)$.
Thus, for all $\epsilon\in(0,1)$, 
\begin{align}
\limsup_{n\to\infty} \frac{1}{n}\log{M_\epsilon^\star(\mathcal{N}_{\sys{A}\to\sys{B}}^{\tensor n})}
& \leq \inf_{\alpha>1} \tilde{I}_\alpha(\mathcal{N}_{\sys{A}\to\sys{B}})
= \multiadjustlimits{
    \inf_{\alpha>1}, 
    \sup_{\rho_{\sys{A}'}\in\set{D}(\hilbert_{\sys{A}'})},
    \inf_{\sigma_\sys{B}\in\set{D}(\hilbert_\sys{B})} }\ \tilde{D}_\alpha\infdiv*{\rho_{\sys{A}'\sys{B}}}{\rho_{\sys{A}'}\tensor\sigma_\sys{B}}
 \nonumber \\
& \leq \multiadjustlimits{
    \inf_{\alpha>1}, 
    \sup_{\rho_{\sys{A}'}\in\set{D}(\hilbert_{\sys{A}'})},
    \inf_{\sigma_\sys{B}\in\set{D}(\hilbert_\sys{B})} }\ \left\{ D\infdiv*{\rho_{\sys{A}'\sys{B}}}{\rho_{\sys{A}'}\tensor\sigma_\sys{B}} + 4(\alpha-1)(\log{v})^2 \right\}
    \label{eq:Dalpha:uniform:bound} \\
& = \sup_{\rho_{\sys{A}'}\in\set{D}(\hilbert_{\sys{A}'})} I(\sys{A}':\sys{B})_{\rho_{\sys{A}'\sys{B}}}, \nonumber
\end{align}
where we use \cite[Lemma~6.3]{tomamichel2012framework} in~\eqref{eq:Dalpha:uniform:bound}, and $v$ is some constant for a given fixed channel $\mathcal{N}_{\sys{A}\to\sys{B}}$.

On the other hand, the asymptotics for the converse bound is relatively straightforward. 
By restricting the supreme over all input density operators $\rho_{{\syss{A}'}_1^n}$ to product states, we have
\begin{align*}
\frac{1}{n}\log{M_\epsilon^\star(\mathcal{N}_{\sys{A}\to\sys{B}}^{\tensor n})} &\geq \frac{1}{n} \adjustlimits \sup_{\rho_{{\syss{A}'}_1^n}\in\set{D}(\hilbert_{{\syss{A}'}_1^n})} \inf_{\sigma_{\syss{B}_1^n}\in\set{D}(\hilbert_{\syss{B}_1^n})} D_{\max}^{\epsilon,{\syss{A}'}_1^n}\infdiv*{\rho_{{\syss{A}'}_1^n\syss{B}_1^n}}{\rho_{{\syss{A}'}_1^n}\tensor\sigma_{\syss{B}_1^n}} \\
&\geq \frac{1}{n} \adjustlimits \sup_{\rho_{\sys{A}'}\in\set{D}(\hilbert_{\sys{A}'})} \inf_{\sigma_{\syss{B}_1^n}\in\set{D}(\hilbert_{\syss{B}_1^n})}  D_{\max}^{\epsilon,{\syss{A}'}_1^n}\infdiv*{\rho_{\sys{A}'\sys{B}}^{\tensor n}}{\rho_{\sys{A}'}^{\tensor n}\tensor\sigma_{\syss{B}_1^n}} \\
&\defas \sup_{\rho_{\sys{A}'}\in\set{D}(\hilbert_{\sys{A}'})} \frac{1}{n}  I_{\max}^{\epsilon}(\dot{{\syss{A}'}_1^n} : \syss{B}_1^n)_{\rho_{\sys{A}'\sys{B}}^{\tensor n}} \tos{$n\to\infty$}{\xrightarrow{\hspace*{2cm}}} \sup_{\rho_{\sys{A}'}\in\set{D}(\hilbert_{\sys{A}'})} I(\sys{A}':\sys{B})_{\rho_{\sys{A}'\sys{B}}},
\end{align*}
where in the last step above, we used the definition of the partial smoothed max-information (see~\cite[Eq.~(11)]{anshu2020partially}) and its asymptotic equipartition property (see~\cite[Eq.~(107)]{anshu2020partially}, also see~\cite[Theorem~6.3]{tomamichel2015quantum}).

Summarizing the above discussion, we have the following theorem.
\begin{theorem}
Let $\mathcal{N}_{\sys{A}\to\sys{B}}$ be a finite-dimensional quantum channel.
For each $\epsilon\in(0,1)$, let $M_\epsilon^\star(\mathcal{N}_{\sys{A}\to\sys{B}})$ denote the smallest alphabet size $M$ such that there exists an $M$-alphabet-size eLOCC protocol that simulates $\mathcal{N}_{\sys{A}\to\sys{B}}$ with fidelity at least $1-\epsilon^2$.
It holds for any $\epsilon\in(0,1)$ that
\begin{equation*}
\lim_{n\to\infty}\frac{1}{n} \log{M_\epsilon^\star}(\mathcal{N}_{\sys{A}\to\sys{B}}^{\tensor n}) = \sup_{\rho_{\sys{A}'}\in\set{D}(\hilbert_{\sys{A}'})} I(\sys{A}':\sys{B})_{\rho_{\sys{A}'\sys{B}}}
\defas C_\fnc{E}(\mathcal{N}_{\sys{A}\to\sys{B}}),
\end{equation*}
where $\rho_{\sys{A}'\sys{B}} \defeq \ID_{\sys{A}'} \tensor \mathcal{N}_{\sys{A}\to\sys{B}}(\braket{\rho}_{\sys{A}'\sys{A}})$, and $\braket{\rho}_{\sys{A}'\sys{A}}\defeq (\sqrt{\rho_{\sys{A}'}}\tensor I_\sys{A}) \braket{\gamma} (\sqrt{\rho_{\sys{A}'}}\tensor I_\sys{A})$.
\end{theorem}

\section*{Acknowledgment}
This research is supported by the National Research Foundation, Prime Minister's Office, Singapore and the Ministry of Education, Singapore under the Research Centres of Excellence programme.
MC and MT are also supported by NUS startup grants (A-0009028-02-00).
RJ is also supported by the NRF grant NRF2021-QEP2-02-P05.
This work was done in part while RJ was visiting the
Technion-Israel Institute of Technology, Haifa, Israel, and the Simons Institute for the Theory of Computing, Berkeley, CA, USA.

\bibliographystyle{IEEEtran}
\bibliography{reference}
\appendices
\section{An Alternative Proof to Lemma~\ref{lemma:f:convex}}
\label{app:proof:lemma:f:convex}
\begin{proof}
Let $J$ and $\tilde{J}$ be the Choi-Jamiolkowski state of the channel $\set{N}_{\sys{A}\to\sys{B}}$ and $\tilde{\set{N}}_{\sys{A}\to\sys{B}}$, respectively, \ie,
\begin{equation*}
    J \defeq \ID_{\sys{A}'}\tensor\mathcal{N}_{\sys{A}\to\sys{B}} \left(\braket{\gamma}_{\sys{A}'\sys{A}}\right), \quad
    \tilde{J} \defeq \ID_{\sys{A}'}\tensor\tilde{\mathcal{N}}_{\sys{A}\to\sys{B}} \left(\braket{\gamma}_{\sys{A}'\sys{A}}\right). 
\end{equation*}
By writing (see~\cite[Eq.~(19)]{ouyang2023approximate})
\begin{align*}
\ID_{\sys{A}'}\tensor\mathcal{N}_{\sys{A}\to\sys{B}}(\braket{\rho}_{\sys{A}'\sys{A}}) &= (\sqrt{\rho}\tensor I) \cdot J \cdot (\sqrt{\rho}\tensor I),\\
\ID_{\sys{A}'}\tensor\tilde{\mathcal{N}}_{\sys{A}\to\sys{B}}(\braket{\rho}_{\sys{A}'\sys{A}}) &= (\sqrt{\rho}\tensor I) \cdot \tilde{J} \cdot (\sqrt{\rho}\tensor I),
\end{align*}
we can rewrite~\eqref{eq:channel:distance} as
\begin{align}
f:(\tilde{\mathcal{N}}_{\sys{A}\to\sys{B}},\rho_\sys{A})&\mapsto\begin{aligned}[t]
\max\ & \frac{1}{2}\tr\left(Z+Z^\dagger\right) \\
\fnc{s.t.}\ & \left(\begin{smallmatrix}(\sqrt{\rho}\tensor I) \cdot J \cdot (\sqrt{\rho}\tensor I) & Z \\ Z^\dagger & (\sqrt{\rho}\tensor I) \cdot \tilde{J} \cdot (\sqrt{\rho}\tensor I)\end{smallmatrix}\right) \geq 0
\end{aligned}\nonumber \\
&= \begin{aligned}[t]
\max\ & \frac{1}{2}\tr\left(Z+Z^\dagger\right) \\
\fnc{s.t.}\ & (\sqrt{\rho}\tensor I) \cdot J \cdot (\sqrt{\rho}\tensor I) \geq Z\cdot (\sqrt{\rho}\tensor I)^{-1} \cdot \tilde{J}^{-1} \cdot (\sqrt{\rho}\tensor I)^{-1} \cdot Z^\dagger
\end{aligned} \nonumber \\
&= \begin{aligned}[t]
\max\ & \frac{1}{2}\tr\left(Z+Z^\dagger\right) \\
\fnc{s.t.}\ &  J \geq (\sqrt{\rho}\tensor I)^{-1} \cdot Z\cdot (\sqrt{\rho}\tensor I)^{-1} \cdot \tilde{J}^{-1} \cdot (\sqrt{\rho}\tensor I)^{-1} \cdot Z^\dagger \cdot (\sqrt{\rho}\tensor I)^{-1}
\end{aligned} \nonumber \\
\label{eq:F:rewrite}
&= \begin{aligned}[t]
\max\ & \frac{1}{2}\tr\left((\rho\tensor I)\cdot (\tilde{Z}+\tilde{Z}^\dagger)\right) \\
\fnc{s.t.}\ & J \geq \tilde{Z} \cdot \tilde{J}^{-1} \cdot \tilde{Z}^\dagger
\end{aligned} ,
\end{align}
where we substitute $\tilde{Z}=(\sqrt{\rho}\tensor I)^{-1} \cdot Z\cdot (\sqrt{\rho}\tensor I)^{-1}$ in the last step.
Note that~\eqref{eq:F:rewrite} is a maximization over linear functions of $\rho$, and therefore much be convex in $\rho$.    
\end{proof}
\end{document}